\documentclass[11pt]{article}

\usepackage{amsmath,amssymb,amsthm,amsfonts,latexsym,bbm,xspace,graphicx,float,mathtools,mathdots,cancel}
\usepackage{braket,caption,subcaption,ellipsis,xcolor,textcomp,MnSymbol}

\usepackage[backref,colorlinks,citecolor=blue,bookmarks=true]{hyperref}
\usepackage[nameinlink]{cleveref}

\usepackage{url}
\usepackage{bibentry} 
\usepackage[letterpaper,margin=1in]{geometry}

\usepackage[explicit]{titlesec}

\usepackage{undertilde}
\newtheorem{theorem}{Theorem}[section]

\newtheorem{corollary}[theorem]{Corollary}%
\newtheorem{definition}[theorem]{Definition}

\newtheorem{proposition}[theorem]{Proposition}

\newtheorem{conjecture}[theorem]{Conjecture}
\newtheorem{problem}[theorem]{Problem}

\newcommand\C{{\mathbb{C}}}

\newcommand\AR{{\utilde{\mathbf{R}}}}
\newcommand\F{{\mathbb{F}}}
\newcommand\sF{{\mathcal{F}}}
\newcommand{\R}{\mathbf{R}}
\newcommand\Z{{\mathbb{Z}}}
\newcommand\Q{{\mathbb{Q}}}
\newcommand\N{{\mathbb{N}}}
\newcommand\poly{{\mathop\textup{poly}}}
\newcommand\chr{{\mathop\textup{char}}}

\newcommand\remove[1]{{}}

\title{A stronger connection between the asymptotic rank conjecture and the set cover conjecture}
\author{Kevin Pratt\thanks{Department of Computer Science, Courant Institute of Mathematical Sciences, New York University.}}

\begin{document}

\maketitle

\begin{abstract}
We give a short proof that Strassen's asymptotic rank conjecture implies that for every $\varepsilon > 0$ there exists a $(3/2^{2/3} + \varepsilon)^n$-time algorithm for set cover on a universe of size $n$ with sets of bounded size. This strengthens and simplifies a recent result of Bj\"orklund and Kaski that Strassen's asymptotic rank conjecture implies that the set cover conjecture is false. From another perspective, we show that the set cover conjecture implies that a particular family of tensors $T_n \in \C^N \otimes \C^N \otimes \C^N$ has asymptotic rank greater than $N^{1.08}$. Furthermore, if one could improve a known upper bound of $\frac{1}{2}8^n$ on the tensor rank of $T_n$ to $\frac{2}{9 \cdot n}8^n$ for any $n$, then the set cover conjecture is false.
\end{abstract}

\section{Introduction}
In a recent preprint \cite{bk}, Bj\"orklund and Kaski have shown that the \emph{set cover conjecture} \cite{cygan,krauthgamer2019set} from the area of exact exponential algorithms, and Strassen's \emph{asymptotic rank conjecture} \cite{strassen1994algebra}, a broad generalization of the conjecture that the exponent of matrix multiplication $\omega$ equals 2, are inconsistent with each other. The purpose of this note is to give a simpler and quantitatively informative proof of this result. 

We begin by recalling the relevant problems.

\begin{problem}[$s$-Set Cover]\label{prob:setcov}
Given $t \in \N$ and  $\sF \subseteq 2^{[n]}$ where each $X \in \sF$ has size at most $s$, decide if there exist at most $t$ sets in $\sF$ whose union equals $[n]$.
\end{problem}
\begin{conjecture}[Set Cover Conjecture \cite{cygan,krauthgamer2019set}]\label{conj:scc}
	For every $\varepsilon > 0$ there exists $s \in \N$ such that there is no (possibly randomized) algorithm for $s$-set cover with runtime $2^{(1-\varepsilon)n}$.
\end{conjecture}

Via standard reductions we will be interested in the following algorithmic problem.

\begin{problem}[Balanced Tripartitioning]\label{prob:bst}
Given $\sF_1, \sF_2, \sF_3 \subseteq \binom{[3n]}{n}$, decide if there exist $S_i \in \sF_i$ with $S_1 \cup S_2 \cup S_3 = [3n]$.
\end{problem}

This can be solved in time $8^n \cdot \poly(n)$ as follows. Let $f_i : \Z_2^{3n} \to \{0,1\}$ be the indicator function of $\sF_i$. Note that $(f_1 * f_2 * f_3)(1^{3n})>0$ if and only if the answer to \Cref{prob:bst} is ``yes" (here $(f * g)(x) := \sum_{y \in \Z_2^{3n}} f(y)g(x-y)$ denotes the convolution of $f$ and $g$). Using Fourier inversion in $\Z_2^{3n}$, we can compute this quantity in $8^n \cdot \poly(n)$ time. We will see that improving significantly on this simple algorithm would refute \Cref{conj:scc} (this fact is essentially shown in \cite{bk} as well).

 Clearly \Cref{prob:setcov} requires at least $\binom{3n}{n} \approx 6.75^n$ time. We show that one can nearly achieve this (that is, solve \Cref{prob:setcov} in time $(6.75 + \varepsilon)^n$ for any $\varepsilon > 0$) if the asymptotic rank conjecture holds. This improves on the main result in \cite{bk}, which roughly showed that \Cref{prob:bst} can be solved in time $7.999^n$ under this assumption.  Similar to \cite{bk}, we will obtain algorithms for \Cref{prob:bst} from algorithms for evaluating certain trilinear forms, or \emph{tensors}. These will in turn arise from hypothetical bounds on the \emph{asymptotic rank} of certain ``base" tensors. Our simplification is due to the fact that the trilinear forms we will be interested in are more symmetric than those in \cite{bk}. This is also related to the fact that the main algorithmic problem studied in \cite{bk} is an ``almost balanced" tripartitioning problem, rather than the ``exactly balanced" \Cref{prob:bst}.

\subsection{Background on asymptotic tensor rank}
We refer the reader to \cite{blaser2013fast} and Chapters 14 and 15 of \cite{burgisser2013algebraic} for background on tensor rank and its applications, but we quickly recap the relevant notions here. Throughout $\F$ denotes an arbitrary field. For our purposes one may assume without loss of generality that $\F = \Q$ or $\F= \Z_p$ (this will be due to \cite[Section 15.3]{burgisser2013algebraic}). We assume access to a $\poly(n)$-time algorithm for arithmetic with $n$-bit elements in the prime field of $\F$. By a tensor we mean a trilinear form $T: \F^n \times \F^n \times \F^n \to \F$, i.e., a cubic set-multilinear polynomial in three disjoint sets of $n$ variables. We say that $T$ has \emph{dimension} $n$. The \emph{support} of $T$ is the set of monomials appearing with nonzero coefficient (with respect to some implicitly chosen basis). We say that $T$ is \emph{concise} if for all choices of bases, every variable is contained in some monomial in the support of $T$. We say that $T$ is \emph{tight} if for some choice of bases of $\F^n$, there exist injective functions $f, g, h : [n] \to \Z$ such that $f(i)+g(j)+h(k) = 0$ for all $(i,j,k)$ in the support of $T$.\footnote{We encourage the reader to ignore the notion of tightness; it is only needed to state Strassen's original conjecture. The more basic question of whether \emph{all} tensors have minimal asymptotic rank \cite[Problem 15.5]{burgisser2013algebraic} is open.}

The following tensors are key.
\begin{definition}
	Let
	\[T_k := \sum_{\substack{S, T, U \in \binom{[3k]}{k}\\ S \cup T  \cup U = [3k]}}X_S Y_T Z_U.\]
\end{definition}

%Our original interest in these trilinear forms was due to the fact that faster algorithms for their exact evaluation would imply faster algorithms for computing the permanent. This follows from similar ideas to those in this note. However, we will see in the next section that faster algorithms for their exact evaluation would violate SETH.

\begin{definition}
Let $T$ be a tensor. We say that $T$ has rank one if $T = \sum a_i X_i \sum b_i Y_i \sum c_iZ_i$ for some $a,b,c \in \F^n$. The tensor rank of $T$, denoted $\R(T)$, is the minimum number $r$ such that $T$ can be expressed as an $\F$-linear combination of $r$ rank-one tensors.
\end{definition}
To orient the reader, it is not difficult to show that every tensor has rank at most $O(n^2)$, and that ``most" tensors meet this bound (see \cite[Theorem 20.9]{burgisser2013algebraic} for precise bounds). It is also not difficult to see that if $T$ is concise, then $\R(T) \ge n$.

%The following is only needed to avoid trivialities in the statement of the asymptotic rank conjecture.
%\begin{definition}
%Let $T = \sum_{i,j,k \in [n]} a_{ijk}X_iY_jZ_k $ be a tensor. We say that $T$ is concise if the bilinear forms 
%\[\sum_{i,j,k \in [n]} a_{ijk}X_iY_{jk}, \sum_{i,j,k \in [n]} a_{ijk}X_{j}Y_{ik},\sum_{i,j,k \in [n]} a_{ijk}X_{k}Y_{ij},
%\]
%have rank $n$.
%\end{definition}
%Note that concise tensors have tensor rank at least $n$.
\begin{definition}
The Kronecker product of $T = \sum_{i,j,k \in [n]} a_{ijk}X_iY_jZ_k $ and $T' = \sum_{i,j,k \in [n]} b_{ijk}X_iY_jZ_k$ is the trilinear form $T \otimes T' : \F^{n^2} \times \F^{n^2} \times \F^{n^2} \to \F$ given by
\[T \otimes T' := \sum_{i,i',j,j',k,k' \in [n]} a_{ijk}b_{i'j'k'}X_{i,i'}Y_{j,j'}Z_{k,k'}.\]
We use the shorthand $T^{\otimes r}$ to denote the $r$-fold Kronecker product of $T$ with itself.
\end{definition}
Note that $\R(T \otimes T') \le \R(T) \R(T')$. As a result, the following is well-defined.
\begin{definition}
$\AR(T) := \lim_{r \to \infty} \R(T^{\otimes r})^{1/r}$ is the \emph{asymptotic rank} of $T$.
\end{definition}
%This follows from a straightforward adaptation of \cite[Proposition 15.1]{burgisser2013algebraic}.
The significance of asymptotic rank is that it that it characterizes the asymptotic complexity of computing high Kronecker powers of $T$. More formally, if $T$ is concise then $\AR(T) \le x$ implies that for every $\varepsilon > 0$, one can compute $T^{\otimes r}$ using $O((x+\varepsilon)^r)$ arithmetic operations.\footnote{This is not true of (non-asymptotic) tensor rank, which only characterizes multiplicative complexity.}  This can be understood as a generalization of the relevance of tensor rank to fast matrix multiplication, and follows from a recursive algorithm analogous to Strassen's algorithm \cite{strassen1969gaussian} and all subsequent improvements thereon.  Furthermore, this is optimal, as the computation of a trilinear form $T$ requires $\Omega(\R(T))$ multiplications \cite[Corollary 6]{baur1983complexity}.\footnote{``Rank" can be replaced by ``border rank" and all of our results still hold, because asymptotic rank and asymptotic border rank are equivalent \cite[Lemma 15.27]{burgisser2013algebraic}.} 

\begin{conjecture}[Asymptotic Rank Conjecture \cite{strassen1994algebra}]\label{conj:arc}
For every tensor $T : \F^n \times \F^n \times \F^n \to \F$ that is tight and concise, $\AR(T) = n$.
\end{conjecture}

That is, this posits that tensors of potentially high rank ``simplify" as much as is possible under taking Kronecker powers. This conjecture is probably difficult, as a positive answer would imply that $\omega = 2$ and a negative answer would necessarily give \emph{explicit} high-rank tensors\footnote{Amusingly, a disproof of the asymptotic rank conjecture could conceivably not tell us what the polynomial-time Turing machine producing the implied family of high-rank tensors is, but only that it exists.}, a longstanding challenge in algebraic complexity. To illustrate how poorly understood the behavior of tensor rank under powering is, as far as we are aware it is consistent with the current state of knowledge that for all $n$-dimensional tensors, $\R(T^{\otimes 2}) \le 5  n^2$ --- in other words, we do not even know if tensors simplify as much as is possible once they are squared! More embarrassingly yet, we are not aware of \emph{any} efficiently computable map from the space of $n$-dimensional tensors to the space of $n^2$-dimensional tensors whose image contains tensors of superlinear (i.e.~$\omega(n^2)$) rank. Finally, we are unaware of any example of tensor with (border) rank greater than $n$ but with asymptotic rank $n$. If no such example exists, then the moral opposite of the asymptotic rank conjecture is true!

\subsection{Results}
Our main result is the following:

\begin{theorem}\label{thm:main}
	For any $k \in \N$ and $\varepsilon > 0$, there is a randomized algorithm for balanced tripartitioning with runtime
	
	\[\left (\frac{(\AR(T_k)+ \varepsilon) \cdot 27^k}{\binom{3k}{k} \binom{2k}{k}} \right )^{n/k}.\]
\end{theorem}

By a reduction from \Cref{prob:bst} to \Cref{prob:setcov} similar to that of \cite[Section 4]{bk}, we then obtain the following.
\begin{corollary}\label{cor:scres}
	If the asymptotic rank conjecture is true, then for every fixed $\varepsilon > 0$ and $s \in \N$, there is a randomized algorithm for $s$-set cover with runtime $(3/2^{2/3}+\varepsilon)^n$.
\end{corollary}
For context, \cite{bk} showed that if the asymptotic rank conjecture is true, then for any $s \in \N$ there is a randomized algorithm for $s$-set cover with runtime $1.99999^n$ (see \cite[Item 5 of Section 1.2]{bk}).

We now highlight some more-or-less immediate consequences of \Cref{thm:main}.

\begin{corollary}\label{cor:arcor}
	If the set cover conjecture is true, then for every $\varepsilon > 0$, for all sufficiently large $k$,  $\AR(T_k) > (8-\varepsilon)^k$.
\end{corollary}

\begin{corollary}\label{cor:ranklb}
	If the set cover conjecture is true, then for every $k$, \[\R(T_k) \ge 8^k \cdot \binom{3k}{k} \binom{2k}{k}/27^k \ge \frac{2}{9} \cdot 8^k \cdot k^{-1}.\]
\end{corollary}

Before proceeding to the proofs of \Cref{thm:main,cor:scres}, we pause to make some comments. 

\begin{enumerate}

	\item Assuming \Cref{conj:arc}, $k = 11$ is the smallest value for which \Cref{thm:main} would beat the $8^n \cdot \poly(n)$-time algorithm for \Cref{prob:bst} sketched earlier. The resulting base tensor is of modest dimension $\binom{33}{11} > 10^8$. The improvement over the trivial algorithm in \cite{bk} followed from a base tensor of dimension $7$.
%	\begin{center}
	%	\includegraphics{fig.pdf}
%	\end{center}

	\item The tensor $T_1$ arose in work of Coppersmith and Winograd \cite[Section 11]{coppersmith1987matrix}, where it was noted that if $\AR(T_1) = 3$, then $\omega = 2$. Note that \Cref{thm:main} does not conflict with this possibility. Moreover, key to \cite{coppersmith1987matrix} is the fact that the induced matching number of the supporting hypergraph of $T_k$ is nearly maximal. See \cite[Appendix A]{fu2014improved} for an exposition of this fact. In the language of \cite{cohn2005group}, these induced matchings are called \emph{uniquely solvable puzzles}.
	
	\item  \Cref{cor:arcor} shows that the set cover conjecture implies the existence of an $N$-dimensional tensor (that is, $T_k$ where $N = \binom{3k}{k})$ with asymptotic rank greater than $N^{1.08}$. On the other hand, it is known that the asymptotic rank of any $N$-dimensional tensor is at most  $N^{2\omega/3}$ \cite[Proposition 3.6]{strassen1988asymptotic}. As $1.09 < 4/3$, this does not imply that $\omega = 2$ is inconsistent with \Cref{conj:scc}.

%	\item The reader may be wondering whether even stronger hypothetical runtimes could be achieved for \Cref{prob:setcov} by using bounds on the asymptotic rank of the natural order-$m$ generalization of $T_k$, call it $T_{k,m}$. This is not the case. By a \emph{flattening} argument \cite[Section 3.4]{landsberg2012tensors} and the fact that disjointness matrices have full rank [], it follows that $\R(T_{k,m}) \ge \binom{km}{k \lfloor m/2 \rfloor} = \binom{N}{N  \lfloor m/2 \rfloor/m}$ where $N := km$. By Stirling's approximation this is roughly $2^{H( \lfloor m/2 \rfloor/m)N}$ where $H$ is the binary entropy function. This is minimized when $m=3$, where it equals $3/2^{2/3}$.
	
	\item The best upper bound we know on $\R(T_k)$ is $8^k/2$, valid for any field with $\chr(\F) \neq 2$. \Cref{cor:ranklb} says that a tantalizingly small improvement on this would violate \Cref{conj:scc}. This rank upper bound is closely related to the Fourier algorithm sketched earlier, and is realized as follows. Let $G$ be an abelian group. Suppose that there is a function $f : \binom{[3k]}{k} \to G$ and $x \in G$ such that $f(S) + f(T) + f(U) = x$ if and only if $S, T, U$ are disjoint. We could then obtain the upper bound $\R(T_k) \le |G|$ by zeroing-out and relabeling variables in the tensor $\sum_{a, b, c \in G , a+b+c=x} X_aY_bZ_c$, which has rank $|G|$ when $\chr(\F) \neq 2$. The most obvious way to instantiate this idea is to take $G = \Z_2^{3k}$, let $f$ be the indicator vector of $S$, and let $x$ be the all-ones vector. But we can do a little better by taking $G = \Z_2^{3k-1}$, letting $f(S)$ be the the indicator vector $\sum_{s \in S} e_s$ if $1 \in S$, and otherwise $f(S) = 1^{3k-1} + \sum_{s \neq 1 \in S} e_i$, and letting $x$ be the all-ones vector.
	
	\item The proof of \Cref{thm:main} will imply that all of these results hold for any tensor with the same support as $T_k$.\footnote{Technically, one needs access to efficient arithmetic in the minimal subfield of $\F$ containing the coefficients of this tensor.}
\end{enumerate}

\section{Proofs of \Cref{thm:main} and \Cref{cor:scres}}
We first note that the tensors $T_k$ satisfy the condition in \Cref{conj:arc}.
\begin{proposition}
	$T_k$ is concise and tight.
\end{proposition}
\begin{proof}
	Conciseness of $T_k$ is equivalent to the saying that the bilinear form $\sum_{S \in \binom{[3k]}{k}, T \in \binom{[3k]}{2k}, S \cup T = [3k]} X_S Y_T$ has rank $\binom{3k}{k}$. This is immediate. To see that $T_k$ is tight, define $f : \binom{[3k]}{k} \to \Z$ by mapping $S$ to $\sum_{i \in S} 4^i$. Then $f$ is injective, and moreover $f(S)+f(T)+f(U) = \sum_{i=1}^{3n} 4^i$ if and only if $S,T,U$ are in the support of $T_k$. So $f' = f - \sum_{i=1}^{3k} 4^i$ witnesses the tightness of $T_k$.
\end{proof}
 
\begin{proof}[Proof of \Cref{thm:main} ]
Given an instance of \Cref{prob:bst} on a universe of size $3n_0$ and with families $\sF_{10}, \sF_{20}, \sF_{30}$, let $n := k \lceil n_0/k \rceil$ be the smallest multiple of $k$ larger than $n_0$, and set $r := n/k$. Pick any equipartition $S_1 \sqcup S_2 \sqcup S_3=[3n] \setminus [3n_0]$. We will consider the instance of \Cref{prob:bst} on $[3n]$ with families $\sF_i := \{ X \cup S_i : X \in \sF_{i0}\}$. This yields an instance of \Cref{prob:bst} which has a solution if and only if the original one did, which is constructible in time $\poly(n)$ (remembering $k$ is a constant), and on a negligibly larger universe (since $n < n_0 + k$).

Let $X  = \{x \in \{0,1\}^{3rk} : |(x_{1+3km}, x_{2+3km}, \ldots, x_{3k+3km})| = k$ for $m = 0, \ldots, r-1 \}$, where $| \cdot |$ denotes the Hamming weight. That is, in each $x \in X$ the $r$ consecutive blocks of length $3k$ contain $k$ ones each. Note that \[T_k^{\otimes r} = \sum_{(a,b,c) \in X^3 : a \vee b \vee c = 1^{3rk}} X_aY_bZ_c.\]
	 
The algorithm works as follows. First choose a permutation $\sigma \in \mathfrak{S}_{3n}$ uniformly at random. Writing $\sigma(\sF) := \{ \sigma(Y) : Y \in \sF \}$, note that $\sigma(\sF_1), \sigma(\sF_2), \sigma(\sF_3)$ contains a tripartition if and only if $\sF_1, \sF_2, \sF_3$ did. Let $\sF_i' := \sigma(\sF_i) \cap X$. If $\{\sF_i\}$ contained no tripartition, certainly neither will $\{\sF_i'\}$. If $\{\sF_i\}$ did contain a tripartition, $\{\sF_i'\}$ also will with probability at least
\[p := \frac{(\binom{3k}{k}\binom{2k}{k})^r}{\binom{3n}{n} \binom{2n}{n}}\]
because a tripartition is sent to a uniformly random tripartition by $\sigma$, and $X$ contains $(\binom{3k}{k}\binom{2k}{k})^r$ (ordered) tripartitions.

Note that if we set all variables $X_a, Y_b, Z_c$ where $a \notin \sF_1', b \notin \sF_2', c \notin \sF_3'$ to zero, the resulting restriction of $T_k^{\otimes r}$ is identically zero if and only if $\{\sF'_i\}$ contained a tripartition. So let $Y$ be a subset of the prime field of $\F$ with $|Y| = 4$ (taking an extension if $\chr(\F) \in \{2,3\}$). Set $X_a = 0$ if $a \notin \sF_1'$ and let $X_a$ be a uniformly random element of $Y$ otherwise, and similarly for the $Y$ and $Z$ variables. By what we just said this is always zero on ``no" instances, and by the Schwartz--Zippel lemma this is nonzero with probability at least $3/4$ on ``yes" instances.\footnote{Note that the use of randomization is unnecessary if $\chr(\F) = 0$.} By assumption, this evaluation can be done using $O((\AR(T_k)+\varepsilon/2)^r)$ field operations. Because asymptotic rank is invariant under field extension \cite[Proposition 15.17]{burgisser2013algebraic}, we may assume that these field operations only involve elements in $Y$ and a constant-sized subset of the prime field of $\F$ arising in a rank decomposition of a fixed power of $T_k$. So this evaluation takes $(\AR(T_k)+\varepsilon/2)^r \cdot \poly(n)$ time.

Repeat this test $1/p$ times and output ``no" just when all tests fail. This always rejects no instances and accepts yes instances with probability at least $1-(1-3p/4)^{1/p} >e^{-3/4}$. The total time taken is
\[ (\AR(T_k)+\varepsilon/2)^r \cdot \frac{\binom{3n}{n} \binom{2n}{n}}{\binom{3k}{k}^r\binom{2k}{k}^r} \cdot \poly(n) < (\AR(T_k)+\varepsilon)^r \cdot \frac{27^{kr}}{\binom{3k}{k}^r\binom{2k}{k}^r}.\qedhere\]

\end{proof}

\begin{proof}[Proof of \Cref{cor:scres} ]
Given set family $\sF$, first construct its downwards closure $\sF' := \cup_{X \in \sF} 2^X$. This is done in time $|\sF| \cdot \poly(n)$. Note that $\sF$ contains at most $t$ sets covering $[n]$ if and only if $\sF'$ contains at most $t$ sets \emph{partitioning} $[n]$. Moreover, notice that for any partition $X_1 \sqcup \cdots \sqcup X_m = [n]$ with $|X_i| \le s$, there exists a partition $A \sqcup B \sqcup C = [m]$ such that $\lfloor n/3 \rfloor - s \le |\sqcup_{a \in A} X_a| \le \lfloor n/3 \rfloor+s$, and similarly with the sets indexed by $B$ and $C$.

Next, for each $m \le t$, compute all unions of all pairwise disjoint collections of $m$ subsets in $\sF'$, having size at most $n/3+s$. This can be done with dynamic programming in time $\binom{n}{n/3} \cdot \poly(n)$. Next remove from the resulting set family all sets of size less than $\lfloor n/3 \rfloor-s$. Let $\sF'_1, \ldots, \sF'_t$ be the resulting set families. 

By the first paragraph, it now suffices to check for every $(t_1,t_2,t_3)$ with $t_1+t_2+t_3 \le t$ if there exist $X \in \sF_{t_1}', Y \in \sF_{t_2}', Z \in \sF_{t_3}'$ partitioning $[n]$. For every $S \subset [n]$ of size $3s$, and for every partition $S = S_1 \sqcup S_2 \sqcup S_3$, for $i \in [3]$ construct  
\[\sF_{t_i, S}'' = \{X \setminus S_i : X \in \sF_{t_i}' , |X| = n/3-s+|S_i|, X \cap S = S_i\}.\]
This is all done na\"ively in time $\binom{n}{n/3} \cdot \poly(n)$. If there were sets in $\sF_{t_1}', \sF_{t_2}', \sF_{t_3}'$ partitioning $[n]$, then there exists an $S$ such that $\sF_{t_1, S}'',\sF_{t_2, S}'', \sF_{t_3, S}''$ contains a balanced tripartition of $[n]\setminus S$, and conversely. The sets $\sF_{t_i, S}''$ have equal size $\lfloor n/3 \rfloor-s$ by construction. We then call the balanced tripartitioning algorithm on the universe $[n] - S$ with set families $\sF_{t_1, S}'', \sF_{t_2, S}'', \sF_{t_3, S}''$.

The number of calls made is bounded by the number of triples $(t_1, t_2, t_3)$ and the number of choices of $S$, which are both $\poly(n)$. If the asymptotic rank conjecture is true, then by \Cref{thm:main}, for any fixed $k$ and $\delta > 0$, each call takes $((\binom{3k}{k} + \delta) 27^k/(\binom{3k}{k} \binom{2k}{k} ))^{n/(3k)} \cdot \poly(n)$ time. As the first factor approaches $(27/4)^{n/3}$ as $k \to \infty$ and $\delta \to 0$, by choosing $k$ sufficiently large and $\delta$ sufficiently small, we conclude. \qedhere 
\end{proof}

\bibliographystyle{amsalpha}
 \bibliography{refs}
\end{document}